\documentclass[sigconf]{acmart-mod}

\usepackage[ruled,vlined]{algorithm2e}
\usepackage{amsthm}

\AtBeginDocument{%
  \providecommand\BibTeX{{%
    \normalfont B\kern-0.5em{\scshape i\kern-0.25em b}\kern-0.8em\TeX}}}

\newtheorem{theorem}{Theorem}[section]

\newtheorem{lemma}{Lemma}[theorem]
\newtheorem{proposition}{Proposition}[theorem]
\theoremstyle{definition}
\newtheorem{definition}{Definition}[section]
\theoremstyle{definition}

\begin{document}

\title{Permutation Encoding for Text Steganography: A Short Tutorial}

\author{George D.\ Monta\~{n}ez}
\email{gmontanez@g.hmc.edu}
\orcid{0000-0002-1333-4611}
\affiliation{%
  \institution{AMISTAD Lab \\Harvey Mudd College}
  \streetaddress{301 Platt Blvd.}
  \city{Claremont}
  \state{CA}
  \country{USA}
  \postcode{91711}
}

\renewcommand{\shortauthors}{Monta\~{n}ez}
\newcommand{\A}{\mathbb{A}}

\begin{abstract}
  We explore a method of encoding secret messages using factoradic numbering of permuted lists of text or numeric elements. Encoding and decoding methods are provided, with code, and key aspects of the correctness of the methods are formally proven. The method of encoding is simple and provides a working example of using textual and numeric lists as a stenagographic channel. Given the ubiquity of lists, such channels are already present but are often unused.
\end{abstract}

\maketitle

\section{Introduction}


In many correctional facilities in the US, all outgoing inmate correspondence is examined and read before being mailed out. Of all the mail leaving one particular facility, Arion's letters were the most scrutinized. Arion was a gang boss who had repeatedly been caught sending orders to outside gang members through his letters. While past codes that used alternating capitalization and other alphabetic word patterns were eventually discovered, a fear had emerged that illegal messages were somehow still escaping officer detection. For one, his orders continued to be carried out; second, Arion's letters had become more numerous. They had also become more mundane, giving lists of his favorite activities, foods, and books, and in one case,  his ``top 15'' gangster movies:
{\small
\begin{enumerate}
\item[1.] AMERICAN GANGSTER
\item[2.] ONCE UPON A TIME IN AMERICA
\item[3.] THE GODFATHER
\item[4.] THE GODFATHER III
\item[5.] CARLITO’S WAY
\item[6.] THE UNTOUCHABLES
\item[7.] GOODFELLAS
\item[8.] GET CARTER
\item[9.] WHITE HEAT
\item[10.] KING OF NEW YORK
\item[11.] PUBLIC ENEMY
\item[12.] A BRONX TALE
\item[13.] DONNIE BRASCO
\item[14.] THE GODFATHER II
\item[15.] SCARFACE
\end{enumerate}
}

Aside from the questionable placement of Scarface at the bottom of the list, nothing about the list appeared out of place; all movies listed were actual gangster movies, and all were well-received, appearing on other ``top'' movie lists online. No code involving equidistant letter patterns could be found; since all movies were capitalized, neither could any code be found involving capitalization. Eventually, a code was discovered. Evidence was found that this list contained instructions, encoded using permutation information of the elements making up the list. Once deciphered, the message was clear and harrowing: \textit{bury him}. Other messages embedded in his previous lists were also uncovered.

Permutations represent a change in state, and what can change states can be used to store information, and by extension, transmit messages. Permutation codes are one way to represent the permutations of a sequence of objects, which have grown in popularity due to their potential application to powerline communications \cite{vinck2011coded,chu2004constructions} and their useful (if limited) error-correction capabilities \cite{makur2020bounds,smith2012new,vinck2011coded}. Permutation codes, like Lehmer codes \cite{diallopermutation}, are constructed to preserve provable minimum distances between code words (leading to their error-correction capabilities \cite{cameron2010permutation,smith2012new,makur2020bounds}), and although studies of their applications and theoretical properties have seen increasing recent research interest \cite{cameron2010permutation,diallopermutation,vinck2011coded,smith2012new,makur2020bounds}, the essential idea of using permutations to store and transmit information dates back to at least 1965 \cite{david1965permutation}. The factorial number systems on which they are built date back to at least the 19th century~\cite{factoradic}.

We concern ourselves here with using permutations as a basic steganographic channel that can be employed whenever a list is given which has some accepted canonical (or otherwise pre-established) total order. While the channel is not cryptographically secure (being trivially decipherable once the algorithm and baseline ordering are known), the fact that innocuous lists are given in many diverse contexts (including rosters, inventory records, or in our previous example, ``top $K$'' lists) implies that such a channel is often already present yet unused. Furthermore, the curse of dimensionality \cite{koppen2000curse} affects the permutation scaling behavior of lists with many elements; given $n$ items there are $n!$ possible permutations of them, and if an arbitrary permutation is chosen as the baseline ordering, the probability of discovering embedded messages using permutation encoding drops precipitously. For example, given 100 elements to permute, there are approximately $9.33\times 10^{157}$ possible permutations, which is larger than the number of subatomic particles in the known universe. Clearly it is infeasible to search through all the possible permutations to find one for which English-language messages appear using the decoding algorithm. Still, we make no strong claims about the security of such an encoding channel (for the sake of argument accepting its insecurity, at least in the case of a known canonical order such as lexicographic ordering), and only note the channel's ready availability and rare usage. 

The ideas presented here are not new (first being presented in \cite{chakinala2006steganographic}, and later being independently rediscovered in \cite{montanez2011information}). We review an improved version of the encoding and decoding algorithms of Monta\~nez \cite{montanez2011information}, based on factorial (factoradic) numbering~\cite{knuthfactorial,factoradic} and originally constructed for Latin alphabet messages, extending the methods to account for arbitrary alphabets. We prove the correctness of key aspects of the algorithms, make some observations concerning the behavior and scaling of the channel, suggest improvements, and highlight some limitations and opportunities for future work. We begin with descriptions of the factoradic encoding and decoding algorithms.

\section{Overview}

The basic idea of the encoding algorithm is that, given an ordered list of elements, it is possible to encode messages based on the ordering of the elements in the list. We begin with an alphabet $\A$, consisting of some finite set of distinguishable elements. As a concrete example, we might consider the 26 Latin alphabetic letters [a--z] plus a space character, allowing us to interpret each nonempty message in $\A^*$ as a base-27 number. Given such a number and a predetermined baseline ordering over the original list of elements, we can map that number to a specific permutation of a list of elements. The conversion process is to first convert the plaintext message to a number, then map that number to its factoradic permutation. To decode, we recover a number from the permuted ordering of items in the list, then convert that base-10 number into a base-$b$ number, which recovers the original plaintext message.

\section{Encoding Algorithm}

Let $x$ be a finite sequence drawn from alphabet $\A$ of size $b = |\A|$, with $0 < b < \infty$. The encoding algorithm consists of the following steps:
    \begin{enumerate}
        \item Convert $x$ into a base-10 number $s$.
        \item Convert $s$ to an ordered list of integers, which can subsequently be used to index a baseline list.
    \end{enumerate}

\subsection{Converting to Base 10} 

We convert $x$ from a base-$b$ number to a base-10 number as follows. Let $x = x_0x_1 \ldots x_{m-1}$ with $x_i$ denoting the $i$th letter of $x$ for $0 \leq i < m$. We calculate $s$, the base-10 representation of $x$, as
\[
    s = \sum_{i=0}^{m-1} b^i \cdot \mbox{index}(x_i, \A)
\]
where $\mbox{index}(x_i, \A)$ denotes the index of element $x_i$ in $\A$.

\subsection{Converting $s$ to a List of Integers} 

We map $s$ to a permuted list of $n$ integers where $n$ is the smallest number such that $n! > s$. The output, $q$, of this process will eventually be used to index into an ordered set of elements (such as a lexicographically-ordered list of gangster movies), giving us our steganographic channel. Let $r = [0,\ldots,n-1]$ represent a list containing the integers $0$ to $n-1$ in ascending order. It is from this list $r$ that we will select indices to output for our permuted list.

At each step of encoding, we must decide which element of $r$ to select for the next index to be added to $q$. The number of possibilities at each step is the number of items remaining in $r$, and at each step we select and remove a single item from the list. Let $s_i$ denote the remaining sum at step $i$ when counting down from $n$ to $1$, with $i = |r|$ and $s_n := s$. Let $d_i$ denote the zero-based index of an item in $r$ that we choose as the next item to remove and place in the output list $q$. We define $d_i$ as
\begin{equation*}
    d_i := \left\lfloor \frac{s_i}{(i - 1)!} \right\rfloor
\end{equation*}
and the remaining sum $s_i$ is then updated at each step as
\begin{align}\label{eq:s-recurrence}
    s_{i-1} &= s_{i} - d_i  \cdot (i - 1)!.
\end{align}
The procedure is repeated until all items are removed from $r$ and placed in the output list. This list can then be used to reorder the elements of the baseline steganographic list, using each integer as an index into the baseline ordering of elements. For example, if the baseline ordered list was
\begin{verbatim}
['A BRONX TALE', 'AMERICAN GANGSTER', 'CARLITO’S WAY']
\end{verbatim}
and the output list of integers $q$ was \verb+[2,0,1]+, we would reorder our steganographic list as 
\begin{verbatim}
['CARLITO’S WAY', 'A BRONX TALE', 'AMERICAN GANGSTER'].
\end{verbatim}

We note that, beginning with $s_0 = 0$, Equation~\ref{eq:s-recurrence} gives us a way to reverse the encoding process and recover $s$ given a permuted sequence of numbers $q$ and a baseline permutation to compare against, since $s_{i-1} = s_{i} - d_i  \cdot (i - 1)!$ implies $s_{i} = s_{i-1} + d_i  \cdot (i - 1)!$. Thus, given $q$, we can uniquely recover $d_i$ at each step of decoding, multiply it against $(i-1)!$, and get the next $s_i$ in our sequence, until we reach $s_n$. After giving a summary of the decoding process, we will prove the correctness of this procedure,  demonstrating that the decoding process correctly recovers the original message (modulo messages with trailing ``zeroes'') and proving $s_0 = 0$, as is required by the decoding process. Algorithm~\ref{alg:ENCODING} gives pseudocode for the encoding method, with an example Python implementation given in the Appendix.

\begin{algorithm}[htbp!]
 \KwData{Plaintext message $x$, alphabet base $b$}
 \KwResult{Permuted list of integers $q$}
 Convert base-$b$ message $x$ to base-10 number $s$\;
 Set $n$ to smallest positive integer such that $n! > s$\;
 Set $r := [0,\ldots,n-1]$, $q := [ ]$, and $s_n := s$\;
 \For{$i=n$ \KwTo $1$}{
  Set $d_i := \left\lfloor \frac{s_{i}}{(i-1)!} \right\rfloor$ \;
  Set $s_{i-1} := s_{i} - d_i \cdot (i-1)!$\;
  Append item $r[d_i]$ to $q$\;
  Remove item $r[d_i]$ from $r$ \;
 }
 \Return{q}
 \caption{Pseudocode for Encoding Algorithm}
 \label{alg:ENCODING}
\end{algorithm}

\section{Decoding Algorithm}

Like the encoding algorithm, the decoding algorithm is also simple. We recover the plaintext message from a list of integers using the following steps:
    \begin{enumerate}
        \item Convert the ordered list of integers into a base-10 number, $s$.
        \item Convert $s$ to a base-$b$ number, i.e., the plaintext message $x$.
    \end{enumerate}

\subsection{Converting List of Integers to Base 10} 

The first step of decoding is to transform the ordered list of integers, $q$, back to a base-10 number. To do so, we first create another ordered list $r$, which is a canonically- (or baseline-) sorted version of the list $q$. We iterate over the items in $q$, beginning at position 0, and find the index of the integer at that position in list $r$. We then multiply that index by $(|r| - 1)!$ and add this to our running total. Lastly, we remove the item from $r$, and continue until we have iterated over all items in list $q$. This effectively reverses the steps of the factoradic encoding process. 

To better understand this reversal, we note that at each step of decoding we have a remaining list of items $r$, with element indices ranging from 0 to $|r| - 1$. During encoding, the index number chosen to output at each step is the maximum such that the index multiplied by $(|r| - 1)!$ is less than or equal to $s_i$. In decoding, we use the number chosen during the encoding step to tell us how many copies of $(|r| - 1)!$ we subtracted from $s_i$ at step $i$, being the maximum possible, and so we calculate this product and add it to our running total. This gives us a greedy method of recovering our base-10 number from the permuted list of integers. Repeating this step for all positions of $q$ allows us to reconstruct the integer $s$.

\subsection{Converting from Base 10 to Base b} 

The decoding algorithm converts $s$ from a base-10 integer to the base-$b$ plaintext message $x$ in the standard manner, by  repeatedly dividing by the conversion base and adding the remainder to output at that step. Pseudocode for decoding is given in Algorithm~\ref{alg:DECODING} below.

\begin{algorithm}[htbp!]
 \KwData{Encoded permutation list $q$, alphabet $A$}
 \KwResult{Plaintext message $x$}
 Copy $q$ to new list $r$\;
 Sort $r$ according to baseline (or canonical) ordering\;
 Set $s_0 := 0$\;
 \For{$i=n$ \KwTo $1$}{
  Set $j :=$ index($q[n-i]$, $r$)\;
  Set $s_{n-i+1} := s_{n-i} + j \cdot (i-1)!$\;
  Remove item $r[j]$ from $r$\;
 }
 Convert base-10 number $s$ to base-$b$ message $x$ using $A$\;
 \Return{x}
 \caption{Pseudocode for Decoding Algorithm}
  \label{alg:DECODING}
\end{algorithm}

\section{Correctness Proofs}

Can we ensure that the encoding and decoding algorithms work as expected? For example, can we ensure that the indices $d_i$ produced by dividing the remaining sum by a factorial and taking the floor will always result in a valid index for the remaining items in our reduced list? Can we prove that $s_0$ will always equal $0$, as is presupposed by the decoding method? Are we sure that the decoding algorithm will always reproduce the original $s_n$ value? In this section, we prove the correctness of several aspects of the factoradic encoding and decoding algorithms, answering all of the aforementioned questions in the affirmative.

We begin by restating a few definitions formally, where $i \in \mathbb{Z}^+$ is some positive integer in what follows.

\begin{definition}[$d_i$]\label{def:DI}
    $d_i := \left\lfloor \frac{s_i}{(i-1)!} \right\rfloor$.
\end{definition}

\begin{definition}[$s_{i-1}$]\label{def:SI-1}
    $s_{i-1} := s_i - d_i \cdot (i-1)!$
\end{definition}

\begin{definition}[$s_{i}$]\label{def:SI}
    $s_{i} := s_{i-1} + d_i \cdot (i-1)!$
\end{definition}

Given the above definitions, we now prove a series of propositions, beginning with a needed lemma.

\begin{lemma}[]\label{lem:S0=0}
    For $s_n \in \mathbb{Z}^+$, $s_0 = 0$.
\end{lemma}
\begin{proof}
Under Definition \ref{def:SI-1}, we have 
\begin{align*}
s_0 &= s_1 - \left\lfloor \frac{s_1}{0!} \right\rfloor 0! = s_1 - \left\lfloor s_1 \right\rfloor.
\end{align*}
Beginning with integer $s_n$, at each $i-1$ we take an integer $s_i$ and subtract another integer $d_i \cdot (i-1)!$ from it, which implies $s_1$ is an integer as well, and thus $\left\lfloor s_1 \right\rfloor = s_1$. Therefore, $s_0 = s_1 - \left\lfloor s_1 \right\rfloor = s_1 - s_1 = 0.$
\end{proof}

\begin{proposition}[Nonnegativity of $s_{i}$]\label{prop:NONNEGATIVITY}
    $s_{i} \geq 0$.
\end{proposition}
\begin{proof}
    For integer $i \geq 0$, we have $s_{i+1}/i! \geq \left\lfloor s_{i+1}/i!\right\rfloor$ which implies
    \begin{align*}
        s_{i+1} &\geq \left\lfloor \frac{s_{i+1}}{i!}\right\rfloor i!, \text{ implying }\\
        s_{i+1} &- \left\lfloor \frac{s_{i+1}}{i!}\right\rfloor i! \geq 0, \text{ and thus }\\
        s_{i} &= s_{i+1} - d_{i+1} \cdot i! \geq 0.
    \end{align*}
\end{proof}

\begin{theorem}[Reversibility of Encoding Process]\label{thm:REVERSIBILITY}
    \[
        s_n = \sum_{i=1}^{n} d_i \cdot(i-1)!.
    \]
\end{theorem}
\begin{proof}
    Definition \ref{def:SI-1} implies $d_i \cdot (i-1)! = s_i - s_{i-1}$, so we have
    \begin{align*}
        \sum_{i=1}^{n} d_i \cdot(i-1)! &= \sum_{i=1}^{n} (s_i - s_{i-1}) \\
            &= \sum_{i=1}^{n} s_i - \sum_{i=1}^{n} s_{i-1} \\
            &= s_n + \sum_{i=1}^{n-1} s_i - \sum_{i=0}^{n-1} s_{i} \\
            &= s_n - s_0 = s_n,
    \end{align*}
    where the final equality follows from Lemma~\ref{lem:S0=0}.
\end{proof}
These results tell us that a necessary condition of the decoding procedure (namely, $s_0 = 0$) holds, that the encoding steps will never produce a negative $s_i$ value, and that any message encoded by the encoding process of Algorithm~\ref{alg:ENCODING} can be decoded using Algorithm~\ref{alg:DECODING}, since each $d_i$ can be obtained from the permuted list $q$. Theorem~\ref{thm:REVERSIBILITY} also shows us that we cannot uniquely recover messages with trailing ``zeroes,'' which, practically speaking, means that if we have $\texttt{a} = 0$ in our alphabet then messages differing only in the number of $\texttt{a}$'s as final characters will map to the same $s_n$ value. Since the trailing characters are equivalent to the high-order digits of our base-$b$ number, this is similar to $0005$, $05$, and $5$ all representing the same number. 

Finally, we prove that the index chosen by $d_i$ is always between 0 and $i-1$, so that we never choose an index outside of our available array indices.
\begin{proposition}[Valid Array Indexing]
 $d_i \in \{0,\ldots, i-1\}$.
\end{proposition}
\begin{proof}
 For the lower bound, Proposition~\ref{prop:NONNEGATIVITY} ensures nonnegativity. For the upper bound, by Definition~\ref{def:SI-1} we have
\begin{align*}
    d_{i-1} &= \left\lfloor \frac{s_{i-1}}{(i-2)!} \right\rfloor = \left\lfloor \frac{s_{i}- d_i\cdot (i-1)!}{(i-2)!} \right\rfloor \\
            &= \left\lfloor \frac{s_{i}}{(i-2)!}- d_i\cdot (i-1) \right\rfloor \\
            &= \left\lfloor \frac{s_{i}(i-1)}{(i-1)!}- \left\lfloor \frac{s_{i}}{(i-1)!} \right\rfloor (i-1) \right\rfloor \\
            &= \left\lfloor (i-1)\left(\frac{s_{i}}{(i-1)!}- \left\lfloor \frac{s_{i}}{(i-1)!} \right\rfloor\right)\right\rfloor.
\end{align*}
Note that because $\frac{s_{i}}{(i-1)!}- \left\lfloor \frac{s_{i}}{(i-1)!} \right\rfloor \leq 1$, we obtain
\begin{align*}
    d_{i-1} &\leq \left\lfloor (i-1) \cdot 1 \right\rfloor = i-1.
\end{align*}
\end{proof}

\section{Channel Properties}

How much information can be transmitted across a permutation encoding channel like the one suggested here? Given $n!$ equally likely permutations to choose from, one can transmit $$\log_2 n! \approx n\log_2 n - n\log_2 e + O(\log_2 n)$$ bits of information when transmitting a message using a permuted list of $n$ elements. For example, a uniformly selected permutation of 10 elements has a capacity of 21.79 bits, which is roughly four letters drawn from a 27 element alphabet (letters plus a space). This assumes that all permutations are equally likely, which will not necessarily hold for the method given above. 

\subsection{Bias in Element at First Position}

\begin{figure}[htbp]
 \includegraphics[width=1.\linewidth]{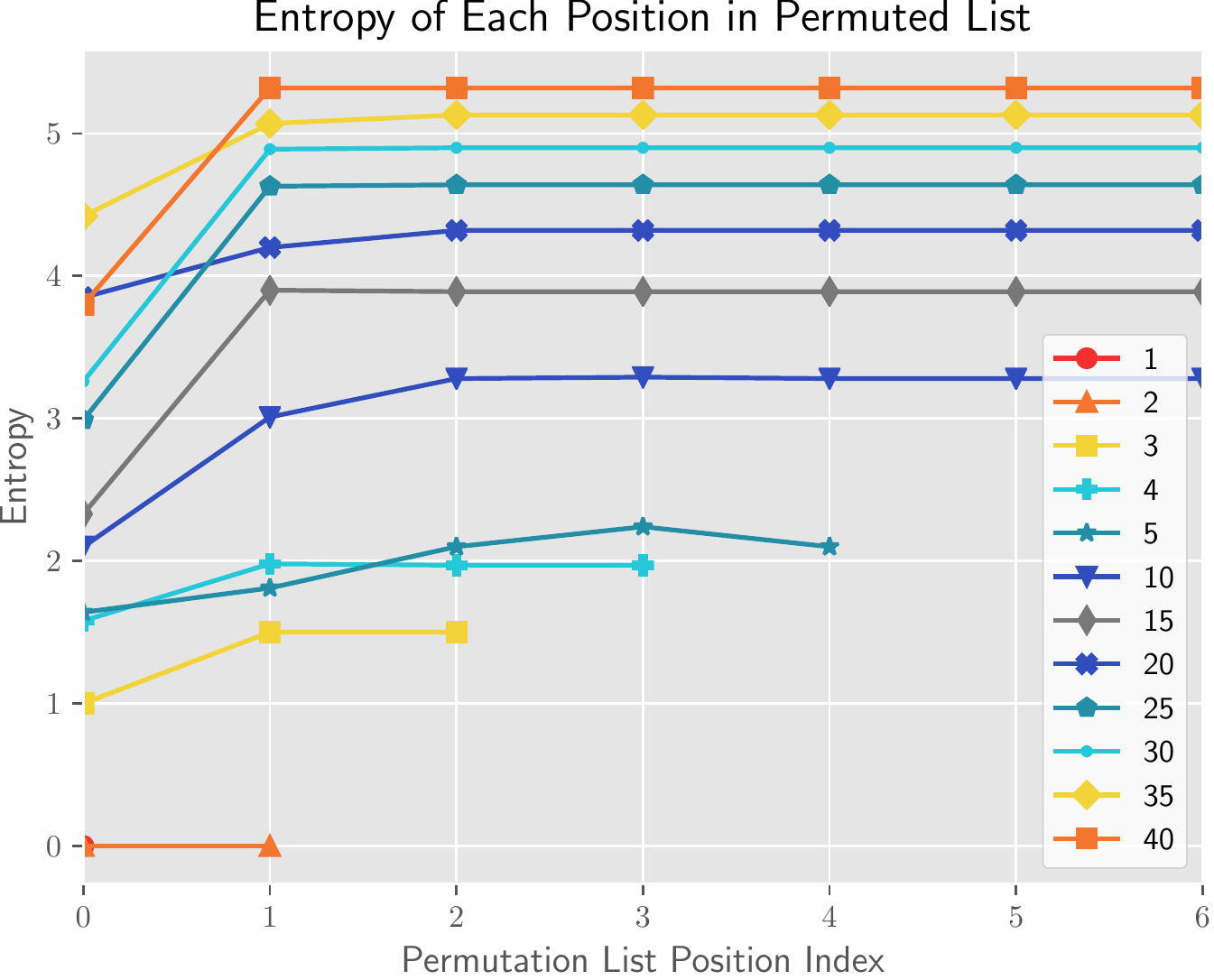}
 \caption{Entropy at each position in permutation list based on the estimated distribution of number frequency at that position.}
  \Description{Entropy at each position in permutation list based on the estimated distribution of item occurrence frequencies at that position.}
  \label{fig:ENTROPY-AT-POSITION}
\end{figure}

Our factoradic encoding method produces biased permutation lists, where the uncertainty of which element is chosen for the first position is reduced relative to the choices for the other positions in the permuted list. Figure~\ref{fig:ENTROPY-AT-POSITION} shows, for various permuted list lengths, what the uncertainty (\textit{entropy}) at each position is. The entropies were computed relative to estimated distributions for element occurrence at each position, based on one million uniformly randomly sampled plaintext messages of each length, from length 1 to length 50. The figure plots the estimated entropy at each position for a subset of lengths, where a line terminates based on the number of elements in the permuted list. For example, the line labeled ``2'' corresponds to permutation lists containing two elements, whereas the line labeled ``10'' corresponds to lists with ten elements. Only the first seven positions are shown, since the lines are roughly horizontal for all remaining positions.

We notice that for lists with one and two elements the entropy is estimated to be zero at all positions, meaning that the same elements are always selected at those positions. While this makes sense for one-element lists (having only one element to select), it is surprising that two element lists behave in the same way. Second, we notice that, with the exception of the first two, in all cases plotted  (along with all those not shown) the entropy of the first position is lower than for all other positions. This implies that the frequency with which certain elements are selected to be output as the first element of the permuted output list is skewed, reducing the uncertainty at that site. For example, if smaller elements are typically chosen for the first position, this will reduce the entropy at that position. Given this bias, how much does it reduce the overall capacity of the channel?

We can gauge the reduction in capacity as follows. For a given permutation list, each position is dependent on the others; selecting a specific element for the $i$th position means it is no longer available to be selected at the $j$th. However, if we consider the average behavior at a list position over many independent trials, we can compute the uncertainty at each position, as was done in Figure~\ref{fig:ENTROPY-AT-POSITION}. We can consider each position as a weighted $|q|$-sided die, where the die is fair if each element of the original lists shows up at this position with equal frequency, in expectation. To the degree that the distribution of elements diverges from uniformity is the degree to which channel capacity will be reduced. 

Figure~\ref{fig:TOTAL-ENTROPY} shows the growth in entropy of the channel (summing the entropy for all positions in a list) as a function of list length, comparing it to the entropy of a channel where each position has a uniform distribution over possible elements (maximizing entropy). For such a uniform channel, we have $n$ options at each position, giving a Shannon surprisal of $\log_2 n$ bits per position, for a total of $n\log_2 n$ bits of entropy. Comparing our observed entropy to this maximum total for the channel, we see that the encoding procedure produces outputs with total entropy very close to the maximum; the largest observed deficit is a mere 2 bits. Thus, the bias of the first position does not significantly reduce the overall message transmission capacity of the channel.

\begin{figure}[htbp]
 \includegraphics[width=1.\linewidth]{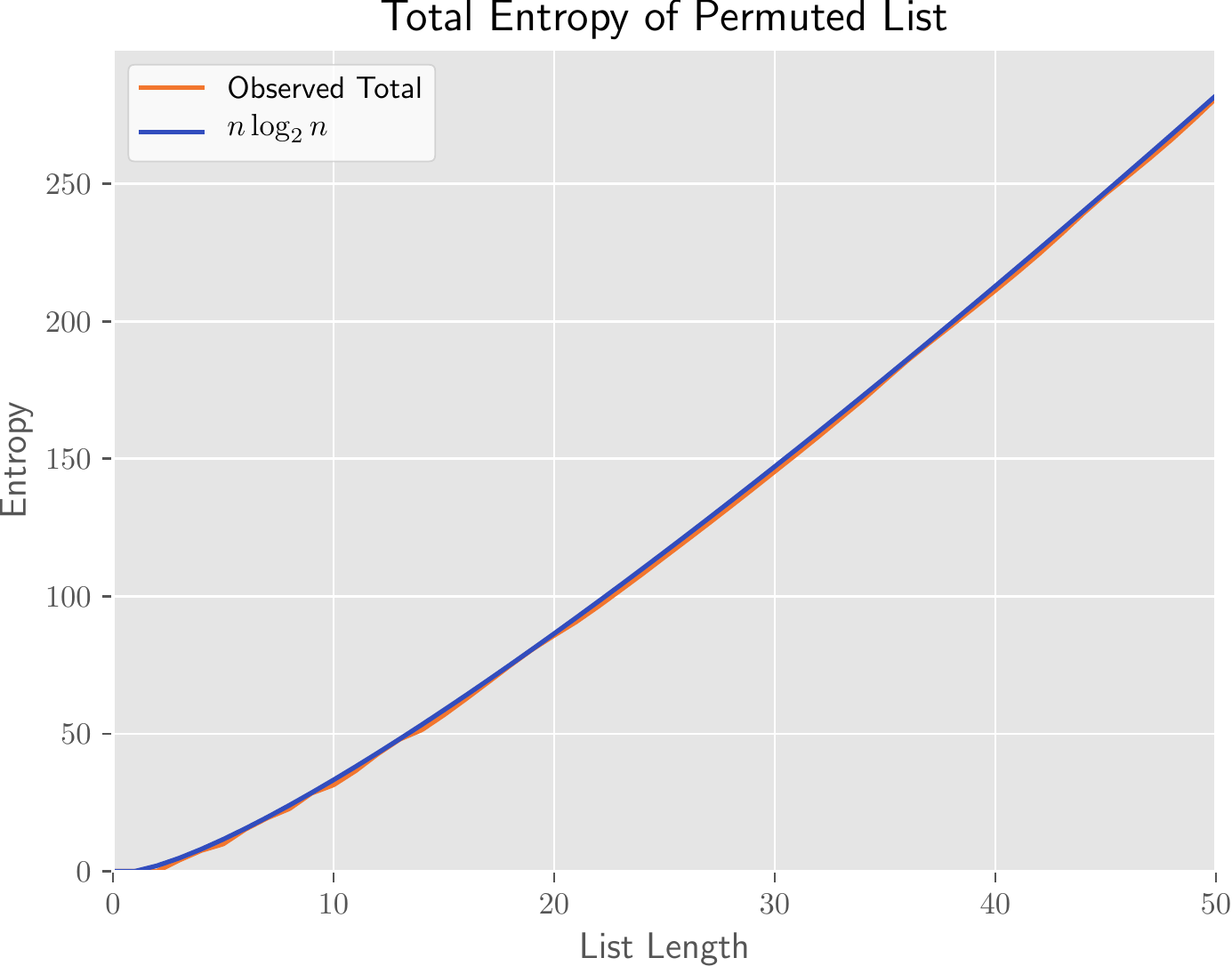}
 \caption{Total observed entropy for permutation lists, summed across all positions, based on the estimated distribution of number frequency at each position.}
  \Description{Total observed entropy for permutation lists, summed across all positions, based on the estimated distribution of number frequency at each position.}
  \label{fig:TOTAL-ENTROPY}
\end{figure}

\subsection{Message Length vs. Permutation Length}

\begin{figure}[htbp]
 \includegraphics[width=1.0\linewidth]{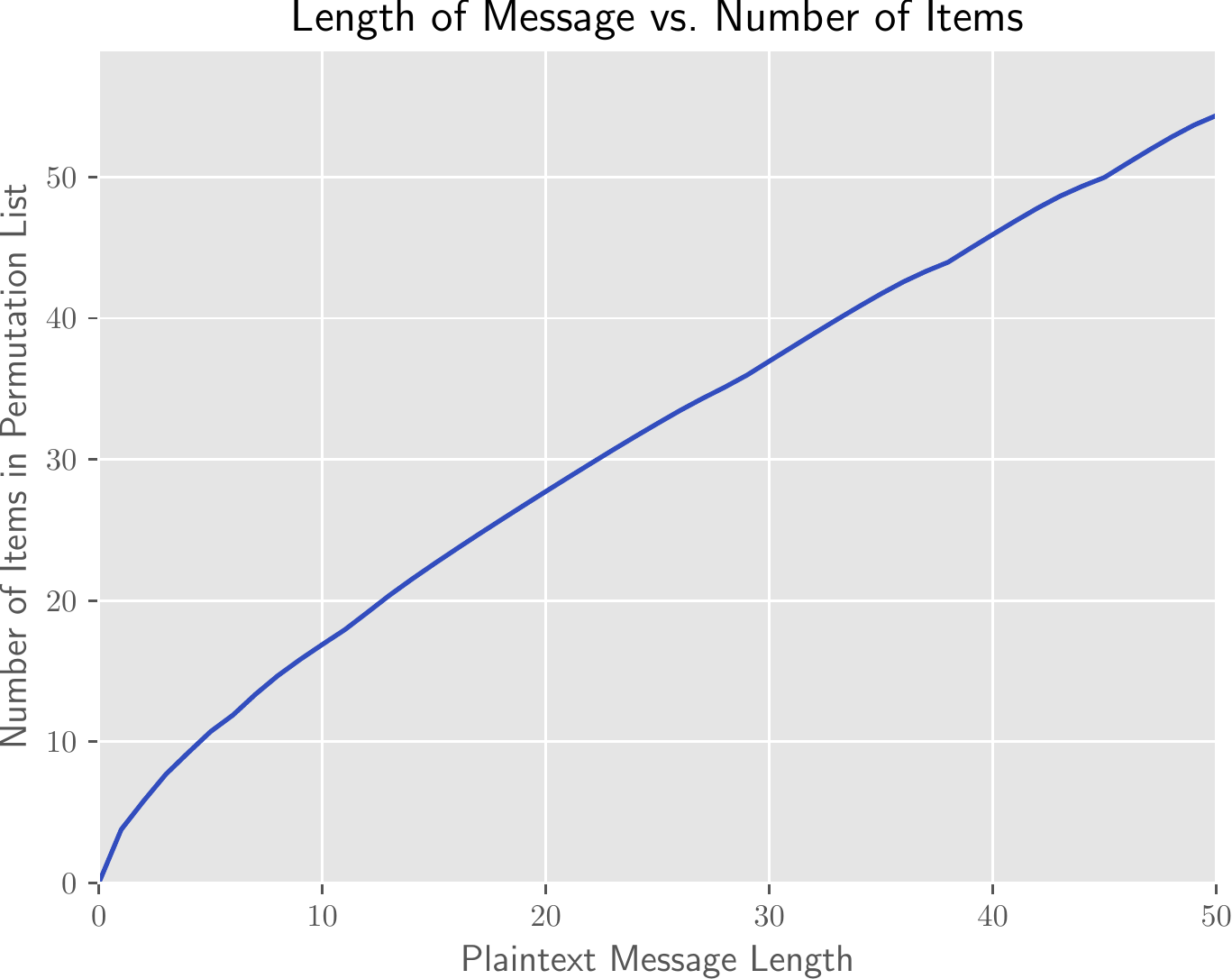}
 \caption{Scaling of plaintext message length versus the estimated mean number of items in the permuted list $q$, surrounded by (imperceptible) 95\% confidence interval.}
  \Description{Estimate of mean number of items for a given plaintext message length, surrounded by a (imperceptible) 95\% confidence interval. Plot was produced from one million uniformly selected random plaintext messages for each message length.}
  \label{fig:MESSAGE-LENGTH}
\end{figure}

Figure~\ref{fig:MESSAGE-LENGTH} shows the estimated mean number of items for a given plaintext message length, surrounded by a 95\% confidence interval. The plot was produced by sampling one million uniformly random plaintext messages for each message length. As seen, after brief super-linear growth, the number of items in the output permutation list needed to encode a given plaintext message becomes close to linear in the length of the message. 

\section{Improvements and Limitations}

\subsection{Optimizing Information Transfer}

Given the unequal likelihood of letters in the English alphabet, one can modify the numbering scheme of mapping alphabetic characters to base-$b$ digits to optimize for assigning small numbers to frequently occurring letters (like `e', `t', and `a'), which will then lower the magnitude of $s$, the equivalent base-10 representation of the plaintext message. A smaller magnitude $s$ requires a smaller list of elements to encode, transmitting more information per element in the list. Such optimizations have not been done here, and are accessible opportunities for future research.

\subsection{Increasing Security}

Given knowledge of the decoding message and the canonical ordering, decoding permuted messages is a trivial task. To increase the security of the steganographic channel one can choose an arbitrary permutation of the list of integers $r$ with which to encode the messages. Trying to decode relative to the wrong permutation will produce gibberish; for example, decoding the enciphered message \texttt{hello} relative to randomly selected permutation \texttt{[5,0,9,10,1,4,6,3,2,8,7]} (rather than the canonical ordering \texttt{[0,$\ldots$,10]}) produces the nonsense string \texttt{cbhwdc} as output. The astronomical number of possible permutations for large $n$ makes a brute-force search for a correct baseline permutation intractable. Thus, introducing a secret and uniformly randomly selected baseline ordering to act as a secret key can add an additional layer of security on top of that provided by the steganographic nature of the channel.

\subsection{Limitations}

As mentioned previously, a limitation of our method is that one cannot ensure unique base-$b$ messages always map to unique $s_n$ values when the message contains trailing characters mapping to the zero-valued element. A simple workaround to is to append a non-zero-valued character to the end of the message before encoding. By providing one particular encoding and decoding method based on factorial number systems, we have a practical example of stenganographic permutation encoding . Undoubtedly, better and more efficient permutation encoding schemes exist, which provide an opportunity for future research.

\section{Conclusion}

Returning to our initial story, you've perhaps identified a plot hole: if Arion has access to a computer on which to run the algorithm, why does he not simply use that computer to send messages directly? While the story is fictional, the method of encoding hidden messages using list of names or common items has been shown to be possible. We have explored one particular method for steganographic permutation encoding, and proven the correctness of key aspects of the encoding and decoding methods presented. In addition to the descriptions of the methods and their pseudocode, example Python implementations are given in Appendix \ref{app:CODE}.

Thinking of permuted lists as a channel within which to transmit messages, rankings of any sort can function as secret channels, as long as they're subject to a total ordering on their elements. For example, student rankings can be used to transmit short messages in a steganographic manner. The same can be done with lists of vehicle identification numbers, phone numbers, license plates, best basketball players of all time, and of course, lists of the best modern gangster movies. 

\bibliographystyle{ACM-Reference-Format}
\bibliography{bibliography}


\begin{thebibliography}{12}


\ifx \showCODEN    \undefined \def \showCODEN     #1{\unskip}     \fi
\ifx \showDOI      \undefined \def \showDOI       #1{#1}\fi
\ifx \showISBNx    \undefined \def \showISBNx     #1{\unskip}     \fi
\ifx \showISBNxiii \undefined \def \showISBNxiii  #1{\unskip}     \fi
\ifx \showISSN     \undefined \def \showISSN      #1{\unskip}     \fi
\ifx \showLCCN     \undefined \def \showLCCN      #1{\unskip}     \fi
\ifx \shownote     \undefined \def \shownote      #1{#1}          \fi
\ifx \showarticletitle \undefined \def \showarticletitle #1{#1}   \fi
\ifx \showURL      \undefined \def \showURL       {\relax}        \fi
\providecommand\bibfield[2]{#2}
\providecommand\bibinfo[2]{#2}
\providecommand\natexlab[1]{#1}
\providecommand\showeprint[2][]{arXiv:#2}

\bibitem[\protect\citeauthoryear{Cameron}{Cameron}{2010}]%
        {cameron2010permutation}
\bibfield{author}{\bibinfo{person}{Peter~J Cameron}.}
  \bibinfo{year}{2010}\natexlab{}.
\newblock \showarticletitle{Permutation codes}.
\newblock \bibinfo{journal}{\emph{European Journal of Combinatorics}}
  \bibinfo{volume}{31}, \bibinfo{number}{2} (\bibinfo{year}{2010}),
  \bibinfo{pages}{482--490}.
\newblock


\bibitem[\protect\citeauthoryear{Chakinala, Kumarasubramanian, Manokaran,
  Noubir, Rangan, and Sundaram}{Chakinala et~al\mbox{.}}{2006}]%
        {chakinala2006steganographic}
\bibfield{author}{\bibinfo{person}{RC Chakinala}, \bibinfo{person}{Abishek
  Kumarasubramanian}, \bibinfo{person}{R Manokaran}, \bibinfo{person}{Guevara
  Noubir}, \bibinfo{person}{C~Pandu Rangan}, {and} \bibinfo{person}{Ravi
  Sundaram}.} \bibinfo{year}{2006}\natexlab{}.
\newblock \showarticletitle{{Steganographic Communication in Ordered
  Channels}}. In \bibinfo{booktitle}{\emph{International Workshop on
  Information Hiding}}. Springer, \bibinfo{pages}{42--57}.
\newblock


\bibitem[\protect\citeauthoryear{Chu, Colbourn, and Dukes}{Chu
  et~al\mbox{.}}{2004}]%
        {chu2004constructions}
\bibfield{author}{\bibinfo{person}{Wensong Chu}, \bibinfo{person}{Charles~J
  Colbourn}, {and} \bibinfo{person}{Peter Dukes}.}
  \bibinfo{year}{2004}\natexlab{}.
\newblock \showarticletitle{Constructions for permutation codes in powerline
  communications}.
\newblock \bibinfo{journal}{\emph{Designs, Codes and Cryptography}}
  \bibinfo{volume}{32}, \bibinfo{number}{1} (\bibinfo{year}{2004}),
  \bibinfo{pages}{51--64}.
\newblock


\bibitem[\protect\citeauthoryear{David}{David}{1965}]%
        {david1965permutation}
\bibfield{author}{\bibinfo{person}{Slepian David}.}
  \bibinfo{year}{1965}\natexlab{}.
\newblock \bibinfo{title}{Permutation code signaling}.
\newblock
\newblock
\newblock
\shownote{US Patent 3,196,351.}


\bibitem[\protect\citeauthoryear{{Diallo, A{\i}ssatou and Zopf, Markus and
  F\"urnkranz, Johannes}}{{Diallo, A{\i}ssatou and Zopf, Markus and
  F\"urnkranz, Johannes}}{2020}]%
        {diallopermutation}
\bibfield{author}{\bibinfo{person}{{Diallo, A{\i}ssatou and Zopf, Markus and
  F\"urnkranz, Johannes}}.} \bibinfo{year}{2020}\natexlab{}.
\newblock \showarticletitle{{Permutation Learning via Lehmer Codes}}.
\newblock \bibinfo{journal}{\emph{24th European Conference on Artificial
  Intelligence (ECAI 2020), Santiago de Compostela, Spain}}
  (\bibinfo{year}{2020}).
\newblock


\bibitem[\protect\citeauthoryear{{Knuth, D. E.}}{{Knuth, D. E.}}{1997}]%
        {knuthfactorial}
\bibfield{author}{\bibinfo{person}{{Knuth, D. E.}}}
  \bibinfo{year}{1997}\natexlab{}.
\newblock In \bibinfo{booktitle}{\emph{{The Art of Computer Programming (3rd
  ed.)}}}, Vol.~\bibinfo{volume}{{2: Seminumerical Algorithms}}.
  \bibinfo{publisher}{{Addison-Wesley}}, \bibinfo{pages}{p. 192}.
\newblock
\newblock
\shownote{ISBN 0-201-89684-2.}


\bibitem[\protect\citeauthoryear{K{\"o}ppen}{K{\"o}ppen}{2000}]%
        {koppen2000curse}
\bibfield{author}{\bibinfo{person}{Mario K{\"o}ppen}.}
  \bibinfo{year}{2000}\natexlab{}.
\newblock \showarticletitle{The curse of dimensionality}. In
  \bibinfo{booktitle}{\emph{5th Online World Conference on Soft Computing in
  Industrial Applications (WSC5)}}, Vol.~\bibinfo{volume}{1}.
  \bibinfo{pages}{4--8}.
\newblock


\bibitem[\protect\citeauthoryear{{Laisant, Charles-Ange}}{{Laisant,
  Charles-Ange}}{1888}]%
        {factoradic}
\bibfield{author}{\bibinfo{person}{{Laisant, Charles-Ange}}.}
  \bibinfo{year}{1888}\natexlab{}.
\newblock \showarticletitle{{Sur la num\'eration factorielle, application aux
  permutations}}.
\newblock \bibinfo{journal}{\emph{{Bulletin de la Soci\'et\'e Math\'ematique de
  France}}}  \bibinfo{volume}{16} (\bibinfo{year}{1888}),
  \bibinfo{pages}{176--183}.
\newblock
\newblock
\shownote{In French.}


\bibitem[\protect\citeauthoryear{Makur}{Makur}{2020}]%
        {makur2020bounds}
\bibfield{author}{\bibinfo{person}{Anuran Makur}.}
  \bibinfo{year}{2020}\natexlab{}.
\newblock \showarticletitle{Bounds on permutation channel capacity}. In
  \bibinfo{booktitle}{\emph{2020 IEEE International Symposium on Information
  Theory (ISIT)}}. IEEE, \bibinfo{pages}{2026--2031}.
\newblock


\bibitem[\protect\citeauthoryear{Monta{\~n}ez}{Monta{\~n}ez}{2011}]%
        {montanez2011information}
\bibfield{author}{\bibinfo{person}{George~D Monta{\~n}ez}.}
  \bibinfo{year}{2011}\natexlab{}.
\newblock \emph{\bibinfo{title}{{Information Storage Capacity of Genetic
  Algorithm Fitness Maps}}}.
\newblock Master's Thesis. \bibinfo{school}{Department of Computer Science,
  Baylor University}.
\newblock


\bibitem[\protect\citeauthoryear{Smith and Montemanni}{Smith and
  Montemanni}{2012}]%
        {smith2012new}
\bibfield{author}{\bibinfo{person}{Derek~H Smith} {and}
  \bibinfo{person}{Roberto Montemanni}.} \bibinfo{year}{2012}\natexlab{}.
\newblock \showarticletitle{A new table of permutation codes}.
\newblock \bibinfo{journal}{\emph{Designs, Codes and Cryptography}}
  \bibinfo{volume}{63}, \bibinfo{number}{2} (\bibinfo{year}{2012}),
  \bibinfo{pages}{241--253}.
\newblock


\bibitem[\protect\citeauthoryear{Vinck}{Vinck}{2011}]%
        {vinck2011coded}
\bibfield{author}{\bibinfo{person}{AJ Vinck}.} \bibinfo{year}{2011}\natexlab{}.
\newblock \showarticletitle{Coded modulation for power line communications}.
\newblock \bibinfo{journal}{\emph{arXiv preprint arXiv:1104.1528}}
  (\bibinfo{year}{2011}).
\newblock


\end{thebibliography}

\appendix

\section{Sample Code}\label{app:CODE}

We give an example implementation in Python of the encoding and decoding algorithms, which assume an alphabet consisting of the Latin alphabetic characters [a--z] plus a space character.

\subsection{encode.py}

The following code requires installation of the \verb+mpmath+ Python module.
{\footnotesize
\begin{verbatim}
import sys
from mpmath import mp, mpf
mp.dps = 1000 

def next_factorial(s):
    i = 1.0
    while s >= 1:
        i += 1
        s = s / mpf(i)
    return int(i)

def facs_table(n):
    facs = [1,]
    for i in range(n - 1):
        facs.append(facs[-1] * (i + 1))
    return facs        

def encode(char_values, alpha_base):
    s = sum([(alpha_base**i) * value for \
            i, value in enumerate(char_values)])
    n = next_factorial(s)
    items = [i for i in range(n)]
    out = []
    facs = facs_table(n)
    for i in range(n):
        d = int(s / facs[-(i + 1)])
        s -= d * facs[-(i + 1)]
        out.append(items[d])
        del items[d]
    return out

def main(content):
    x = content.strip()
    alpha = "abcdefghijklmnopqrstuvwxyz "[:]
    character_values = [alpha.index(letter) \
            for letter in x]
    print(encode(character_values, len(alpha)))

if __name__ == "__main__":
    main(sys.stdin.read())
\end{verbatim}
}

\subsection{decode.py} 

The following code requires installation of the \verb+mpmath+ Python module. 
{\footnotesize
\begin{verbatim}
import sys
from mpmath import mp, mpf
mp.dps = 1000

def facs_table(n):
    facs = [1,]
    for i in range(n - 1):
        facs.append(facs[-1] * (i + 1))
    return facs  

def decode(items, alphabet):
    alpha_base = len(alphabet)
    out = ""
    r = items[:]
    r.sort()
    facs = facs_table(len(items))
    s = 0
    for i, item in enumerate(items):
        s += r.index(item) * facs[-(i + 1)]
        r.remove(item)
    while s > 0:
        part = int(s % mpf(alpha_base))
        out += alphabet[part]
        s = int(s / mpf(alpha_base))
    return out

def main(content):    
    items = [int(item) for item in \
         content.strip(" []\n ").split(",")]
    alpha = "abcdefghijklmnopqrstuvwxyz "[:]
    print(decode(items, alpha))

if __name__ == "__main__":
    main(sys.stdin.read())
\end{verbatim}
}
\subsection{Usage}

To use the encoder code, simply run encode.py from the command line using the Python interpreter (code written for versions $\geq$ 2.*.*), and feed into standard input the phrase you'd like encoded.

\begin{verbatim}
\> echo hello world|python encode.py
input: hello world
output: [1,17,13,5,4,0,3,12,8,15,14,11,16,7,9,10,2,6]
\end{verbatim}

To decode, run the decode.py script and pass in the list of integers to standard input, in the same format as given by the encode method, which will recover the plaintext message.

\begin{verbatim}
\> echo [3,8,6,5,1,7,0,4,10,2,12,9,11]|python decode.py
input: [3,8,6,5,1,7,0,4,10,2,12,9,11]
output: test me
\end{verbatim}
\end{document}